\documentclass[12pt]{amsart}

\usepackage{amsmath,amsthm,amsfonts,amssymb,latexsym,verbatim,bbm,hyperref}
\usepackage[shortlabels]{enumitem}
\usepackage[top=1.5in,bottom=1.5in,left=1.5in,right=1.5in]{geometry}
\usepackage{subcaption}
\usepackage{afterpage}

\DeclareRobustCommand{\SkipTocEntry}[5]{}

\allowdisplaybreaks

\setlist{itemsep=.5\baselineskip,topsep=.5\baselineskip}

\numberwithin{equation}{section}
\theoremstyle{plain}
\newtheorem{theorem}{Theorem}[section]
\newtheorem{thm}[theorem]{Theorem}
\newtheorem{lemma}[theorem]{Lemma}

\newtheorem{prop}[theorem]{Proposition}
\newtheorem{example}[theorem]{Example}

\newtheorem{defn}[theorem]{Definition}

\newtheorem{cor}[theorem]{Corollary}

\newtheorem{rmk}[theorem]{Remark}

\newcommand{\arr}{\rightarrow}
\newcommand{\incl}{\hookrightarrow}

\newcommand{\R}{\mathbb{R}}
\newcommand{\C}{\mathbb{C}}
\newcommand{\Z}{\mathbb{Z}}

\newcommand{\eps}{\epsilon}

\newcommand{\Id}{\mathbbm{1}}

\newcommand{\mcF}{\mathcal{F}}
\newcommand{\mcG}{\mathcal{G}}

\newcommand{\mcS}{\mathcal{S}}

\newcommand{\mcU}{\mathcal{U}}
\newcommand{\mcV}{\mathcal{V}}

\newcommand{\mbN}{\mathbb{N}}

\usepackage{color}

\newtheorem{claim}[theorem]{Claim}

\newcommand{\beq}{\begin{eqnarray}}
\newcommand{\eeq}{\end{eqnarray}}

\newcommand{\Tr}{\mbox{\rm Tr}}

\newcommand{\setft}[1]{\mathrm{#1}}

\newcommand{\Unitary}{\setft{U}}

\usepackage{braket}

\newcommand{\norm}[1]{\left\lVert#1\right\rVert}
\newcommand{\Free}{\mcF}

\renewcommand{\Im}{\operatorname{Im}}
\DeclareMathOperator{\hlp}{hlp}

\DeclareMathOperator{\tr}{tr}
\newcommand{\ntr}{\widetilde{\operatorname{tr}}}
\DeclareMathOperator{\BS}{BS}

\title[Entanglement and hyperlinear profile]{Entanglement in non-local games
and the hyperlinear profile of groups}

\author[William Slofstra]{William Slofstra$^*$}
\thanks{${}^*$Institute for Quantum Computing and Department of
    Pure Mathematics, University of Waterloo, Waterloo, Canada. email:
    \texttt{weslofst@uwaterloo.ca}}

\author[Thomas Vidick]{Thomas Vidick$^\dagger$}
\thanks{${}^{\dagger}$Department of Computing and Mathematical Sciences,
    California Institute of Technology, Pasadena, USA. email:
    \texttt{vidick@cms.caltech.edu}.}

\begin{document}

\begin{abstract}
    We relate the amount of entanglement required to play linear-system
    non-local games near-optimally to the hyperlinear profile of
    finitely-presented groups. By calculating the hyperlinear profile of a
    certain group, we give an example of a finite non-local game for which the
    amount of entanglement required to play $\eps$-optimally is at least
    $\Omega(1/\eps^k)$, for some $k>0$. Since this function approaches infinity
    as $\eps$ approaches zero, this provides a quantitative version of a
    theorem of the first author. 
\end{abstract}

\maketitle

\section{Introduction}

Let $\mcG$ be a two-player non-local game. Such a game is specified by finite sets
of inputs $X$ and $Y$ and outputs $A$ and $B$, associated with the first
(Alice) and second (Bob) players respectively, a probability distribution $\pi$
on input pairs $(x,y)\in X\times Y$, and a winning predicate $V \in
\{0,1\}^{A\times B \times X \times Y}$. A quantum strategy $\mcS$ for the game
$\mcG$ is specified by Hilbert spaces $H_A$ and $H_B$, a state $\ket{\psi} \in
H_A \otimes H_B$, a positive operator-valued measure $\{A_x^a\}_{a\in A}$ on $H_A$ 
for every $x \in X$, and a positive operator-valued measure $\{B_y^b\}_{b \in B}$ on
$H_B$ for every $y \in Y$. A strategy is finite-dimensional if at least one of
$\dim(H_A)$ or $\dim(H_B)$ is finite; in this case we say that the strategy has
dimension $\min\{\dim(H_A),\dim(H_B)\}$. The winning probability of strategy
$\mcS$ in game $\mcG$ is defined as 
\begin{equation*}
    \omega(\mcG;\mcS)\,=\, \sum_{(x,y) \in X\times Y} \,\pi(x,y)\,
        \sum_{(a,b)\in A\times B}\, V(a,b|x,y) \,\bra{\psi} A_x^a \otimes B_y^b \ket{\psi}\;.
\end{equation*}
The quantum value $\omega^q(\mcG)$ of $\mcG$ is the
supremum of winning probabilities $\omega(\mcG;\mcS)$ across all
finite-dimensional quantum strategies $\mcS$. A basic question about $\mcG$ is:
what is the minimum amount of entanglement $E(\mcG,\eps)$ used by a strategy
that achieves winning probability at least $\omega^q(\mcG) - \eps$? Here we
measure entanglement by the Schmidt rank. When lower bounds on $E(\mcG,\eps)$
are known, $\mcG$ can potentially be used to certify entangled states. As a
result, many lower bound results
exist~\cite{PWPVJ08,JPPVW10,JP11,slofstra2011lower,Coladangelo16,CN16,chao2017test,JOP16,NV17,ostrev2016entanglement}. In
particular, it is known that $E(\mcG,\eps)$ can be arbitrarily large. For
example, a result of Ostrev and the second author \cite{ostrev2016entanglement}
states that for any $n \geq 1$, there is a two-player game $\mcG_n$ with input
sets of size $n$ and output sets of size $2$, such that
\begin{equation}\label{eq:ov}
    E\Big(\mcG_n,O\Big(\frac{1}{n^{5/2}}\Big)\Big) \geq 2^{\Omega(\sqrt{n})}\;.
\end{equation}
As with all known results of this type, the lower bound requires games of increasing
size to get $E(\mcG_n,\eps) \arr +\infty$. A recent result of the first author
is that there is a fixed, constant-size non-local game $\mcG$ such that
$E(\mcG,\eps) \arr +\infty$ as $\eps \arr 0$ \cite{Sl17}. In other words, there
is a game which cannot be played optimally using any finite-dimensional Hilbert
space. The purpose of this paper is to prove the following quantitative version
of this result: 
\begin{theorem}\label{T:main}
    Let $E(\mcG,\eps)$ be the smallest possible integer $d$ such that there is
    a quantum strategy $\mcS$ of dimension $d$ for $\mcG$ with success
    probability $\omega(\mcG;\mcS) \geq \omega^q(\mcG) - \eps$. Then there is a
    non-local game $\mcG$ and constants $C,C' > 0$ such that
    \begin{equation*}
        \frac{C}{\eps^{1/6}} \leq E(\mcG,\eps) \leq \frac{C'}{\eps^{1/2}}
    \end{equation*}
    for all $\eps \geq 0$.
\end{theorem}
The game $\mcG$ constructed in the proof of Theorem \ref{T:main} is similar to
the game constructed in \cite{Sl17}, and in particular is an example of a
linear system non-local game with $\omega^q(\mcG) = 1$. Linear system games are
a subclass of non-local games for which the existence of perfect quantum
strategies is controlled by an associated finitely-presented group $\Gamma$,
called the solution group. The proof of Theorem \ref{T:main} is based on the
observation that $\eps$-perfect strategies for linear system games $\mcG$
correspond to approximate representations of the solution group $\Gamma$ of
$\mcG$.
As a result, the function $E(\mcG,\eps)$ is linked to the dimension of
approximate representations of the solution group $\Gamma$ of $\mcG$. For sofic
groups, the asymptotic dimension growth of sofic approximations is measured by
the sofic profile of the group, a concept introduced by Cornulier \cite{Co13}.
Theorem \ref{T:main} can be thought of as a calculation of a ``hyperlinear
profile'' of the solution group $\Gamma$, measuring the asymptotic dimension
growth of unitary approximations to $\Gamma$.

Since hyperlinear profiles do not seem to have been studied heavily before,
we explore some of their properties here. In particular, we give two related
definitions of hyperlinear profile. The first, which we work with throughout
the paper, is defined for finitely-presented groups, and is convenient for
working with non-local games. The second matches Cornulier's definition of
sofic profile, and in particular, is independent of the choice of presentation.
Although we do not know if the two definitions are precisely the same, we show
that they are the same under a natural equivalence relation.

In light of~\eqref{eq:ov}, we do not expect the inverse polynomial scaling from
Theorem \ref{T:main} to be optimal, and we hope that this initial lower bound
opens the door to further results. We find it interesting that the theorem
provides a finite test that may be executed on two spatially isolated
quantum systems, such that the higher the success in the test, the larger the
dimension that can be certified. The correspondence between near-optimal
strategies for linear system games and approximate representations of solution
groups may also be of independent interest. In particular, this correspondence
implies that any $\eps$-optimal strategy can be turned into an
$O(\eps^{1/2})$-optimal strategy with a maximally entangled state (see
Remark~\ref{rk:connes}).

The rest of the paper is organized as follows. In Sections \ref{S:hlprofile}
and \ref{S:bounds}, we introduce our first definition of hyperlinear profile,
and prove our main lower bound. In Section \ref{S:linear}, we recall the notion
of a solution group of a linear system game, and explain how hyperlinear
profile is related to entanglement for strategies with maximally entangled
states. In section \ref{S:generalstates}, we show that hyperlinear profile is
related to entanglement for strategies with general states. In Section
\ref{S:embedding}, we use the embedding theorem of \cite{Sl17} and the results
of Section \ref{S:bounds} to prove explicit bounds on the hyperlinear profile
of the solution group, finishing the proof of Theorem \ref{T:main}. Finally, in
Section \ref{S:presindep} we give the second definition of hyperlinear profile,
and compare the two definitions.

\subsection{Notation}

We use the following notation throughout the paper.  $\Free(S)$ is the free
group generated by $S$, $\mcU(\C^d)$ is the unitary group of $\C^d$, and
$M_d(\C)$ is the set of $d \times d$ matrices.  We use the following 
norms on $M_d(\C)$: the operator norm $\norm{\cdot}_{op}$, the Frobenius norm
$\norm{\cdot}_F$ defined by $\norm{A}_F = \sqrt{\Tr(A^* A)}$, and the
normalized Frobenius norm $\norm{\cdot}_f = \norm{\cdot}_F / \sqrt{d}$.  Given
a positive semidefinite matrix $\rho$, we also let $\norm{\cdot}_{\rho}$ denote
the seminorm defined by $\norm{A}_{\rho} = \sqrt{\Tr(A^* A \rho)}$.  Note that
$\norm{\cdot}_F = \norm{\cdot}_{\Id}$, $\norm{\cdot}_f = \norm{\cdot}_{\Id/d}$,
and $\norm{A}_{\rho} = \norm{A\rho^{1/2}}_F$. We use $\ntr$ for the linear
functional $\Tr(\cdot) / d$ on $M_d(\C)$. 

\subsection{Acknowledgements}

We are indebted to Narutaka Ozawa for suggesting the use of the Connes
embedding trick and the beautiful line of argument now incorporated in Section
\ref{S:generalstates}; this lead to a substantial improvement in our results.
The first author also thanks Martino Lupini for helpful discussions. 

The second author is supported by NSF CAREER Grant CCF-1553477, AFOSR YIP award
number FA9550-16-1-0495, a CIFAR Azrieli Global Scholar award, and the IQIM, an
NSF Physics Frontiers Center (NSF Grant PHY-1125565) with support of the Gordon
and Betty Moore Foundation (GBMF-12500028). 

\section{Hyperlinear profile of finitely-presented groups}\label{S:hlprofile}

In this section, we state our first definition of hyperlinear profile, along
with some basic properties. The starting point is the following definition
from, e.g., \cite{Sl17} or \cite{HS17}.
\begin{defn}[\cite{Sl17}]
    An $\eps$-representation of a finitely-presented group $G = \langle S : R
    \rangle$ is a homomorphism $\phi : \Free(S) \arr \mcU(\C^d)$ from the free
    group $\Free(S)$ generated by $S$ to the unitary group $\mcU(\C^d)$, such
    that 
    \begin{equation*}
        \norm{\phi(r) - \Id}_f \leq \eps
    \end{equation*}
    for all $r \in R$. 

    An element $g \in G$ is \emph{non-trivial in approximate representations}
    if there is some representative $w \in \Free(S)$ of $g$ and $\delta > 0$
    such that for all $\eps > 0$, there is an $\eps$-representation $\phi$ with
    $\norm{\phi(w)-\Id}_f \geq \delta$. 
\end{defn}
Suppose $\langle S : R \rangle$ and $\langle S' : R' \rangle$ are two
presentations for a group $G$. While the set of $\eps$-representations depends
on the choice of presentation, any $\eps$-representation with respect to
$\langle S : R \rangle$ is an $O(\eps)$-representation with respect to $\langle
S' : R'\rangle$ (with the constant depending on the isomorphism between the
two presentations), and vice-versa. Similarly, whether $g$ is non-trivial in approximate
representations is independent of the choice of presentation.  The dependence
on $\delta$ is also somewhat arbitrary, due to the following well-known
consequence of the tensor-power trick:
\begin{lemma}\label{L:tensorpower}
    Suppose $g_1,\ldots,g_n \in G = \langle S : R \rangle$ are all non-trivial
    in approximate representations. Let $w_i \in \Free(S)$ be a representative
    of $g_i$ for each $1 \leq i \leq n$, and choose $\delta \in (0,\sqrt{2})$.
    Then for every $\eps > 0$, there is an $\eps$-representation $\phi$ such
    that $\norm{\phi(w_i) - \Id}_f \geq \delta$ for all $1 \leq i \leq n$.
\end{lemma}
A finitely-presented group is said to be \emph{hyperlinear} if every
non-trivial element is non-trivial in approximate representations. It is not
known if there is a group which is not hyperlinear, and deciding this is a
major problem in the field. For comparison, if every non-trivial element is
non-trivial in finite-dimensional representations, then the group is said to be
\emph{residually-finite}, and there are examples of finitely-presented,
hyperlinear, but non-residually-finite groups. Suppose $G$ is such a group,
so that there is $w \in \Free(S)$ representing an element which is trivial in all
finite-dimensional representations, but non-trivial in approximate
representations.  If we fix $\delta > 0$, then the dimension of
$\eps$-representations $\phi$ with $\norm{\phi(w) - \Id}_f \geq \delta$ must
increase as $\eps \arr 0$.  The \emph{hyperlinear profile} of $G$ is a
collection of functions which measure these growth rates. 
\begin{defn}\label{D:hlp1}
    Let $G = \langle S : R \rangle$ be a finitely-presented group, and let $T$
    be a finite subset of $\Free(S)$. The \emph{(hyperlinear) profile}
    of $T$ is the function $\hlp(T) : \R_{>0} \times \R_{>0} \arr
    \mbN \cup \{+\infty\}$ such that $\hlp(T;\delta,\eps)$ is the smallest
    integer $d$ for which there is an $\eps$-representation $\phi$ of dimension
    $d$ with 
    \begin{equation*}
        \norm{\phi(w) - \Id}_f \geq \delta \quad\text{ for all } w \in T,
    \end{equation*}
    or $+\infty$ if no such $d$ exists.

    If $T = \{w\}$, then we write $\hlp(w;\delta,\eps)$ for
    $\hlp(T;\delta,\eps)$. The hyperlinear profile of $G$ is the collection of
    functions $\hlp(T)$, where $T$ is a finite subset of $\Free(S)$ not
    containing any element with trivial image in $G$. 
\end{defn}
If $T$ does contain an element with trivial image, then $\hlp(T;\delta,\eps)$
will always be infinite for small enough $\eps$. It is clear that $\hlp(T)$ is
non-decreasing in $\delta$ and non-increasing in $\eps$.  Another easy property
of the hyperlinear profile is that it is non-decreasing under homomorphisms, in
the following sense:
\begin{lemma}\label{L:hlpprop2}
    Let $G_i = \langle S_i : R_i \rangle$, $i=1,2$, be two finitely-presented
    groups, and suppose $\phi : \Free(S_1) \arr \Free(S_2)$ is a homomorphism
    which descends to a homomorphism $G_1 \arr G_2$. Then there is a constant
    $C \geq 1$ such that 
    \begin{equation*}
        \hlp(T, \delta, C\eps) \leq \hlp(\phi(T),\delta,\eps)\;,
    \end{equation*}
    for any finite set $T \subset \Free(S_1)$ and $\eps,\delta > 0$.
\end{lemma}
Note that while the elements $\phi(R_1)$ are trivial in $G_2$, they
do not necessarily appear in $R_2$. The constant $C$ in Lemma \ref{L:hlpprop2}
depends on how many times the relations from $R_2$ must be applied to show
that the elements $\phi(R_1)$ are trivial. 

We are mainly interested in the asymptotic behaviour of $\hlp(T)$ as $\eps \arr
0$ with $\delta$ fixed.  When $\phi$ induces an isomorphism between $G_1$ and
$G_2$, Lemma \ref{L:hlpprop2} shows that the asymptotics of the functions
$\hlp(T)$ are somewhat independent of the choice of presentation. We can say
the same thing about the choice of representatives in $\Free(S)$: 
\begin{lemma}\label{L:hlpprop3}
    Suppose $w_0,w_1 \in \Free(S)$ represent the same element in $G = \langle
    S : R \rangle$. Then there is a constant $C \geq 1$ such that
    \begin{equation*}
        \hlp(w_0,\delta - C \eps, \eps) \leq \hlp(w_1,\delta,\eps)
    \end{equation*}
    for all $\eps > 0$ and $\delta > C \eps$.
\end{lemma}

Although we won't use it, the following lemma puts some limits on how much the
choice of $\delta$ can affect the asymptotics:
\begin{lemma}\label{L:hlpdelta}
    Let $G = \langle S : R \rangle$, and $0 < \delta < \delta' < \sqrt{2}$.
    Then there is a constant $k \geq 1$ such that 
    \begin{equation*}
        \hlp(T,\delta, \eps) \leq \hlp(T, \delta', \eps) \leq k \hlp(T,\delta,\eps/k)^k
    \end{equation*}
    for all $\eps > 0$ and finite subsets $T \subset \Free(S)$. 
\end{lemma}
For instance, the choice of $\delta$ within the range $(0,\sqrt{2})$ does
not affect whether $\hlp$ is polynomial or exponential in $1/\eps$. 
\begin{proof}
    We only need to prove the second inequality.  The proof is similar to Lemma
    \ref{L:tensorpower}, in that we can use the tensor power trick. Suppose
    $\phi$ is an $\eps$-representation of dimension $d$ with 
    \begin{equation*}
        \norm{\phi(w) - \Id}_f \geq \delta
    \end{equation*}
    for all $w \in T$. Let $\psi$ be the direct sum of $\phi$ with its complex
    conjugate, and with $2d$ copies of the trivial representation. Then $\psi$
    is an $\eps/2$-representation of dimension $4d$, such that $\ntr(\psi(w))$
    is real and non-negative for all $w \in \Free(w)$, and $\norm{\psi(w) -
    \Id}_f \geq \delta / 2$ for all $w \in T$. Since $\norm{\psi(w) - \Id}^2_f
    = 2 - 2 \ntr(\psi(w))$, we have that $\ntr(\psi(w)) \leq 1 - \delta^2 / 8$.
    Hence
    \begin{equation*}
        \norm{\psi(w)^{\otimes n} - \Id}_f^2 = 2 - 2 \ntr(\psi(w))^n \geq
            2 - 2 \left(1 - \frac{\delta^2}{8}\right)^n \;.
    \end{equation*}
    By choosing $n$ large enough, we can make the right-hand side larger than $(\delta')^2$.
    Since $\psi^{\otimes n}$ is a $(n\eps/2)$-representation of dimension $(4 d)^n$,
    we conclude that
    \begin{equation*}
        \hlp\left(T, \delta', \frac{n\eps}{2}\right) \leq 4^n \hlp(T,\delta,\eps)^n,
    \end{equation*}
    and the lemma follows. 
\end{proof}
    
Note that we have not restricted Definition \ref{D:hlp1} to hyperlinear groups.
In terms of hyperlinear profile, a finitely-presented group $G = \langle S : R
\rangle$ is hyperlinear if and only if $\hlp(T,\delta,\eps) < +\infty$ for all
finite subsets $T \subset \Free(S)$ not containing any elements with trivial
image in $G$, and real numbers $\delta \in (0,\sqrt{2})$, $\eps > 0$.
Similarly, $G$ is residually finite if and only if $\hlp(T,\delta,\eps)$ is
bounded as $\eps \arr 0$ for all finite subsets $T \in \Free(S)$ (again, not
containing any elements with trivial image in $G$) and $\delta \in
(0,\sqrt{2})$. 

Although we can choose any $\delta$ in $(0,\sqrt{2})$ when measuring
$\hlp(T,\delta,\eps)$, it can make sense to take $\delta$ to be greater than or
equal to $\sqrt{2}$. The largest value of $\norm{U - \Id}_f$ when $U$ is a
unitary is $2$, and this is achieved when $U = -\Id$. Thus, $\hlp(w,2,\eps)$
measures the growth rate of $\eps$-representations $\phi$ where $\phi(w) =
-\Id$. The following proposition shows that when $w \in \Free(S)$ represents a
central involution in $G$ (which will be the case in the main example of the
next section, and also when working with linear system games), $\delta = 2$ is a natural choice. 
\begin{prop}\label{P:centralinv}
    Suppose $g \in G$ is a central involution, and $G = \langle S : R \rangle$
    is a presentation such that $g$ has a representative $w \in S$. Let $0 <
    \delta \leq 2$.  Then there is a constant $C \geq 1$ (depending on
    $\delta$) such that
    \begin{equation*}
        \hlp(w,2,C\eps) \leq \hlp(w,\delta,\eps) \leq \hlp(w,2,\eps)
    \end{equation*}
    for all $\eps > 0$. 
\end{prop}
The proof of Proposition \ref{P:centralinv} relies on some simple stability properties. 
\begin{lemma}\label{L:stability1} \ 
    \begin{enumerate}[(a)]
        \item Suppose $X$ is a matrix with $\norm{X}_{op} \leq 1$. Then there
            is a unitary matrix $U$ with $\norm{X - U}_f \leq \norm{X^* X -
            \Id}_f$, and in fact $U$ can be any unitary in a polar decomposition $X = UP$ of $X$.  

        \item If $X, Y$ are unitary matrices and $X^2=\Id$, then there is a
            unitary matrix $Z$ such that $ZX = XZ$ and 
            \begin{equation*}
                \norm{Y-Z}_f \leq \norm{XY - YX}_f.
            \end{equation*}

        \item If $X$ is any normal matrix, there is a self-adjoint matrix $Z$
            with $Z^2 = \Id$ such that
            \begin{equation*}
                \norm{X - Z}_f \leq 2 \norm{X^2 - \Id}_f.
            \end{equation*}
    \end{enumerate}
\end{lemma}
We expect that Lemma \ref{L:stability1} is well-known to experts. We give the proof
of parts (a) and (b) for the convenience of the reader. A proof of part (c) can be 
found in \cite[Lemma 3.6]{Sl17}.
\begin{proof}[Proof of Lemma \ref{L:stability1}]
    For part (a), let $X = U P$ be a polar decomposition of $X$, where $U$ is unitary.
    Then
    \begin{equation*}
        \norm{X-U}_f = \norm{P-\Id}_f = \big\|\sqrt{X^* X} - \Id\big\|_f \leq \norm{X^* X -\Id}_f,
    \end{equation*}
    where the last inequality comes from the fact that $1-z \leq 1-z^2$ for all $z \in [0,1]$. 

    For part (b), let $Z_0 = \frac{1}{2}(Y + X Y X)$. Then $Z_0$ and $X$
    commute, and $\norm{Y-Z_0}_f$ and $\norm{Z_0^* Z_0 - \Id}_f$ are both at
    most $\frac{1}{2}\norm{XY - YX}_f$. So part (b) follows from applying part
    (a) to $Z_0$. 
\end{proof}

We also need that small perturbations of $\eps$-representations 
remain $O(\eps)$-representations.
        
\begin{lemma}[\cite{Sl17}, Lemma 2.3]\label{L:hom-stable}
    Let $\phi$ be a $d$-dimensional $\eps$-representation of the
    finitely-presented group $G = \langle S : R \rangle$.
    If $\psi : \Free(S) \arr \mcU(\C^d)$ is a homomorphism such that 
    \begin{equation*}
        \norm{\psi(s) - \phi(s)}_f \leq \eps'
    \end{equation*}
    for all $s \in S$, then $\psi$ is an $\left(\eps +
    O(\eps')\right)$-representation, where the constant depends on the length
    of the relations in $R$. 
\end{lemma}

\begin{proof}[Proof of Proposition \ref{P:centralinv}]
    We only need to prove the first inequality. Suppose $\phi$ is a $d$-dimensional
    $\eps$-representation of $G$ with $\norm{\phi(w) - \Id}_f \geq \delta$. 
    Then by Lemmas \ref{L:stability1} and \ref{L:hom-stable}, there is an
    $O(\eps)$-representation $\psi$ such that $\psi(w)$ is a central involution in
    the group generated by $\psi(S)$, and $\norm{\phi(s) - \psi(s)} \leq
    O(\eps)$ for all $s \in S$. If $\eps$ is small enough such that
    $\norm{\phi(w)-\psi(w)} \leq \delta/2$, we then have $\norm{\psi(w) -
    \Id}_f \geq \delta / 2$.  This inequality implies that the dimension of the
    $(-1)$-eigenspace of $\psi(w)$ is bounded below by $\Omega(d \delta^2)$. If $P$
    is the projection on the $(-1)$-eigenspace of $\psi(w)$, then it is not
    hard to see that $P \psi P$ is an $O(\eps/\delta)$-representation of $G$
    (see for instance the proof of Lemma 3.9 of \cite{Sl17}). Since $\delta$ is
    fixed, there is a constant $C$ such that for every $d$-dimensional
    $\eps$-representation, there is a $C\eps$-representation $\psi$ of
    dimension at most $d$ with $\psi(w) = -\Id$. 
\end{proof}

\section{Bounds on hyperlinear profile for a specific group}\label{S:bounds}

The question of whether there is a non-hyperlinear group is notoriously
difficult, and it seems reasonable to look for related questions which might be
more approachable. In this context, it seems  natural to seek examples
of groups for which the hyperlinear profile grows particularly fast. We note
that very little seems to be known about this question.  In the case of the sofic
profile, any non-residually-finite group has a sofic profile at least as large
as $1 / \eps$, essentially because the smallest positive value of the normalized Hamming metric on
the symmetric group $S_n$ is $1/n$ \cite[Fact 3.8]{Co13}.
In contrast, for the hyperlinear profile, it is not clear that there is even a
``default'' lower bound of this form.  The only explicit lower bound we are
aware of prior to our work is the following example due to Tobias Fritz
\cite{Fr13}:
\begin{example}
    Let $G$ be the Baumslag-Solitar group $\BS(2,3)$, so 
    \begin{equation*}
        G = \langle u,v : v^{-1} u^2 v = u^3 \rangle.
    \end{equation*}
    It is shown in \cite{Fr13} that there is a constant $C>0$ such that if
    $\phi : \Free(u,v) \arr \mcU(\C^d)$ is a homomorphism with
    \begin{equation*}
        \norm{\phi(v)^{-1} \phi(u)^2 \phi(v) - \phi(u)^3}_{op} \leq \eps,
    \end{equation*}
    for some $\eps >0$, then
    \begin{equation*}
        \norm{\phi(w) - \Id}_{op} \leq e^{Cd^2} \eps,
    \end{equation*}
    where $w = uv^{-1} u v u^{-1} v^{-1} u^{-1} v$. Note that $w \neq e$ in
    $G$. If $\phi$ is an $\eps$-representation of $G$ then
    \begin{equation*}
        \norm{\phi(v)^{-1} \phi(u)^2 \phi(v) - \phi(u)^3}_{op} \leq \sqrt{d} \norm{\phi(v)^{-1} \phi(u)^2 \phi(v) - \phi(u)^3}_{f}
            \leq \sqrt{d} \eps.
    \end{equation*}
    Consequently
    \begin{equation*}
        \norm{\phi(w)-\Id}_f \leq \norm{\phi(w)-\Id}_{op} \leq e^{Cd^2} \sqrt{d} \eps,
    \end{equation*}     
    from which by rearranging terms we get the lower bound
    \begin{equation*}
        \hlp(w,\delta,\eps) = \Omega\big(\sqrt{\ln(\delta/\eps)}\big)
    \end{equation*}
    for any $\delta > 0$. 
\end{example}

The main result of this section, and the key result of this paper, is that
there is a group with a hyperlinear profile somewhere between $1/\eps^{2/3}$
and $1/\eps$, up to constants.
 
\begin{prop}\label{P:bounds}
    Define
    \begin{align*} 
        K = \langle a,b,c,x,y\ :\ & x y x^{-1} = y^2, xcx^{-1} = c,\\ & yay^{-1} = b, y
            b y^{-1} = a, \\ & c=ab, a^2 = b^2 = c^2 = e \rangle.
    \end{align*}
    For any $0 < \delta \leq 2$ there exists $C,C'>0$ such that 
    \begin{equation}\label{eq:Kbounds}
         \frac{C}{{\eps}^{\frac{2}{3}}} \leq \hlp(c,\delta,\eps) \leq \frac{C'}{\eps}
    \end{equation}
    for all $\eps > 0$.
\end{prop}
For the remainder of this section, we let $K$ denote the finitely-presented
group defined in Proposition \ref{P:bounds}. The group is similar to the group defined
in \cite[Section 5]{Sl17}. In particular, $K$ is sofic, and hence hyperlinear.
The element $c$ is a central involution which is trivial in all
finite-dimensional representations, so $K$ is non-residually finite.  Because
we have chosen a presentation for $K$ in which $c$ is a generator, Proposition
\ref{P:centralinv} implies that we only need to prove Proposition
\ref{P:bounds} when $\delta = 2$. In fact, it is not necessary to make $c$ a
generator---we could replace $c$ with $ab$, the relation $c=ab$ with $[a,b]=e$,
and remove $c$ from the presentation---but we keep the slightly redundant
presentation for simplicity. 

The proof of Proposition \ref{P:bounds} is split into several steps. 
For both the upper and lower bound, it is helpful to look at the subgroup 
\begin{equation*}
    K_0 = \langle y,a,b,c : yay^{-1} = b, yby^{-1} = a, c=ab, a^2 = b^2 = c^2 = e\rangle
\end{equation*}
of $K$. Abstractly, $K_0$ is the semidirect product $\Z \ltimes \Z_2 \times
\Z_2$ of $\Z_2 \times \Z_2$ by the automorphism switching the order of the
factors.  Given $U \in \mcU(\C^d)$, we can define a $2d$-dimensional
representation $\phi$ of $K_0$ by
\begin{equation}\label{E:K0rep}
    \phi(a) = \begin{pmatrix} \Id & 0 \\ 0 & -\Id \end{pmatrix}, 
    \phi(b) = \begin{pmatrix} -\Id & 0 \\ 0 & \Id \end{pmatrix},
    \text{ and } \phi(y) = \begin{pmatrix} 0 & \Id \\ U & 0 \end{pmatrix}.
\end{equation}
The first lemma needed for the lower bound in Proposition \ref{P:bounds} is
that every $d$-dimensional $\eps$-representation of $K$ can be turned into a $2d$-dimensional 
$O(\eps)$-representation which restricts to a representation of $K_0$ that satisfies~\eqref{E:K0rep}.

\begin{lemma}\label{L:stability2}
    Let $\phi$ be a $d$-dimensional $\eps$-representation of $K$ such that
    $\phi(c) = -\Id$. Then there is a $2d$-dimensional
    $O(\eps)$-representation $\psi$ of $K_0$ such that
    \begin{equation*}
        \psi(a) = \begin{pmatrix} \Id & 0 \\ 0 & -\Id \end{pmatrix},\quad 
        \psi(b) = \begin{pmatrix} -\Id & 0 \\ 0 & \Id \end{pmatrix},\quad
        \text{ and }\quad \psi(y) = \begin{pmatrix} 0 & \Id \\ U & 0 \end{pmatrix},
    \end{equation*}
    where the blocks are of size $d \times d$, and $U \in \mcU(\C^d)$. 
\end{lemma}
\begin{proof}
    Since $\norm{\phi(a)^2 - \Id}_f \leq \eps$, part (c) of Lemma
    \ref{L:stability1} implies that there is a unitary matrix $Z$ such that $\norm{\phi(a)-Z}_f \leq 2\eps$ and $Z^2 = \Id$. By Lemma \ref{L:hom-stable}, replacing
    $\phi(a)$ with $Z$ and $\phi(b)$ with $-Z$ yields an
    $O(\eps)$-representation $\alpha$ with $\alpha(a)^2 = \alpha(b)^2 = \Id$
    and $\alpha(a) \alpha(b) = -\Id = \alpha(c)$. Let $\beta = \alpha \oplus \widetilde{\alpha}$,
    where $\widetilde{\alpha}$ is the $O(\eps)$-representation defined from $\alpha$ by 
    switching $a$ and $b$, so 
    \begin{equation*}
        \widetilde{\alpha}(x) = \alpha(x), \widetilde{\alpha}(y) = \alpha(y), \widetilde{\alpha}(a) = \alpha(b),
        \widetilde{\alpha}(b) = \alpha(a), \text{ and } \widetilde{\alpha}(c) = -\Id. 
    \end{equation*}
    Then $\beta(a)^2 = \beta(b)^2 = \Id$, $\beta(a)\beta(b) = -\Id =\beta(c)$, and
    $\tr(\beta(a)) = 0$, so there is a choice of basis such that
    \begin{equation*}
        \beta(a) = \begin{pmatrix} \Id & 0 \\ 0 & -\Id \end{pmatrix} \quad\text{ and }\quad 
        \beta(b) = \begin{pmatrix} -\Id & 0 \\ 0 & \Id \end{pmatrix},
    \end{equation*}
    where the blocks are of size $d\times d$. 
		
		\begin{claim}\label{claim:y-block} Let 
		$$B = \begin{pmatrix} Y_1 & Y_2 \\ Y_3 & Y_4 \end{pmatrix} $$
		be a $2d\times 2d$ unitary matrix such that $\norm{B \beta(a)-\beta(b)B}_f \leq \eps'$, for some $\eps'\geq 0$. Then there are $d\times d$ unitaries $V_2$ and $V_3$ such that 
		$$ \norm{B - \begin{pmatrix} 0 & V_2 \\ V_3 & 0 \end{pmatrix}}_f \,=\,O(\eps'). $$
		\end{claim}
		
		\begin{proof}
		Let
    \begin{equation*}
        Y = \begin{pmatrix} 0 & Y_2 \\ Y_3 & 0 \end{pmatrix}.
    \end{equation*}
    The assumption made in the claim implies that $\norm{Y-B}_f\leq \eps'/2$, and since $B$ is assumed unitary, $\norm{Y^* Y - \Id}_f \leq O(\eps')$ (to avoid having
    to add $(\eps')^2$ to the right side of this inequality, note that we can
    assume $\eps' \leq 2$). It follows that $\norm{Y_i^* Y_i - \Id}_f \leq O(\eps')$ 
    for $i=2,3$. Since $B$ is unitary, $\norm{Y_i}_{op} \leq 1$ for $i=2,3$,
    and hence by part (a) of Lemma \ref{L:stability1} there are unitaries
    $V_i$ with $\norm{Y_i - V_i}_f \leq O(\eps')$, $i=2,3$. 
		\end{proof}
		
		Let
    \begin{equation*}
        \beta(y) = \begin{pmatrix} Y_1 & Y_2 \\ Y_3 & Y_4 \end{pmatrix}.
    \end{equation*}
		Applying Claim~\ref{claim:y-block} to $B=\beta(y)$, there exists unitaries  $V_2$ and $V_3$ such that 
		$$ \norm{B - \begin{pmatrix} 0 & V_2 \\ V_3 & 0 \end{pmatrix}}_f \,=\,O(\eps). $$
		Applying Lemma \ref{L:hom-stable}
    once again, we conclude that the homomorphism $\gamma : \Free(\{a,b,c,x,y\}) \arr \mcU(\C^{2d})$ defined by
    \begin{equation*}
        \gamma(y) = \begin{pmatrix} 0 & V_2 \\ V_3 & 0 \end{pmatrix}, \gamma(s) = \beta(s) \text{ for } s \in \{a,b,c,x\}
    \end{equation*}
    is an $O(\eps)$-representation. Then
    \begin{equation*}
        \psi = \begin{pmatrix} V_2^* & 0 \\ 0 & \Id \end{pmatrix} \cdot \gamma
            \cdot \begin{pmatrix} V_2 & 0 \\ 0 & \Id \end{pmatrix}
    \end{equation*}
    satisfies the conditions of the lemma with $U = V_3 V_2$.
\end{proof}

With Lemma \ref{L:stability2}, we can prove the lower bound in~\eqref{eq:Kbounds}. 
\begin{lemma}\label{L:K-lbound}
    There is a constant $C$ such that 
    \begin{equation}\label{eq:K-lbound}
     \hlp(c,2,\eps) \geq \frac{C}{\eps^{\frac{2}{3}}} \;, 
    \end{equation}
    for all $\eps > 0$. 
\end{lemma}
\begin{proof}
Let $\phi$ be a $d$-dimensional $\eps$-representation of $K$ such that $\norm{\phi(c)-\Id}_f \geq 2$, i.e. $\phi(c)=-\Id$.  
    Applying Lemma~\ref{L:stability2} we deduce the
    existence of a  unitary $U\in\Unitary(\C^{2d})$, and unitaries
    $X=\psi(x),Y=\psi(y) \in \Unitary(\C^{2d})$ such that
    \begin{equation}\label{eq:operator-bound}
    Y = \begin{pmatrix} 0 & \Id \\ U & 0 \end{pmatrix}\qquad\text{and}\qquad \big\|XYX^\dagger - Y^2\big\|_{op} \,\leq\, C_1 \, \sqrt{d}\, \eps,
    \end{equation}
    for some constant $C_1>0$ (the second condition uses $\|\cdot\|_{op} \leq
    \sqrt{d}\|\cdot\|_f$). Let $\eps_1 = C_1 \sqrt{d} \eps$. We show that satisfying both
    conditions requires $d  \geq C_2\eps_1^{-1}$ for some constant $C_2>0$; this
    will prove the lemma. 

    Let $\Theta \subseteq \R$ be such that the elements $(\theta\bmod 1)$ are pairwise distinct for $\theta\in\Theta$, and the set of eigenvalues of $U$ is $\Lambda =
    \{ e^{2i\pi \theta},\,\theta\in\Theta\}$. We prove a lower bound on $d$ by showing that the set $\Theta$ must be large. 
		
		For $a,b\in\R$, let $d(a,b)$ measure the arc-length distance between $x=e^{2i\pi a}$ and $y=e^{2i\pi b}$. We first prove a simple lemma that will facilitate manipulation of angles. 
		
		\begin{claim}\label{claim:sq-root}
	Let $a,b\in \R$, and set $x=e^{2i\pi a}$ and $y=e^{2i\pi b}$. Then $d(a,b)\leq \pi|x-y|$. Moreover, for any square root $x' = e^{2i\pi a'}$ of $x$ there is a square root $y'=2^{2i\pi b'}$ of $y$ such that $d(a',b')= \frac{1}{2} d(a,b)$. 
		\end{claim}

\begin{proof}
The statement of the lemma is invariant under rotation, so we may assume $a=0$, $x=1$, and $b \in (-1/2,1/2]$. If $|x-y| \geq 1$ the bound on $d(a,b)$ is trivial, so for the first part of the lemma we can assume in particular that $b\in[-1/4,1/4)$. Then 
\begin{align*}
d(b,0) &= 2\pi |b| \\
&\leq \pi|\sin(2\pi b)| \\
&\leq \pi |1-y|.
\end{align*} 
For the moreover part of the lemma, note that the square roots $x'=e^{2i\pi a'}$ of $x$ are such that $a' \in \frac{1}{2}\Z$. For any such square root, $y'=e^{2i\pi b'}$ for $b' = a' + \frac{1}{2}b$ is a square root of $y$ such that $d(b',a')=d(b/2,0)=d(b,0)/2$. 
\end{proof}

The next claim shows that as soon as $\Theta\neq \emptyset$ it must contain many distinct values. 

    \begin{claim}
        Let $\theta \in \Theta$. For every integer $\ell \geq 0$ there exists angles 
            $$\Big\{\theta_{j}^{(\ell)},\, j\in \{1,\ldots,2^\ell\}\Big\} \subseteq \Theta$$
        such that for each $j\in \{1,\ldots,2^\ell\}$, 
        \begin{equation}\label{eq:ind-bound}
           d\big(\theta_j^{(\ell)},\,2^{-\ell} \theta + j2^{-\ell}\big) \,\leq\, 2\pi\eps_1.
        \end{equation} 
  \end{claim}
	
    \begin{proof}
		If $\eps_1 \geq 1$ the claim is trivial, so assume $\eps_1<1$. 
        We prove the claim by induction on $\ell$.  The base case $\ell=0$ is obtained
        by setting $\theta_j^{(0)} = \theta$. Assume the statement of the claim true
        for some $\ell \geq 0$. Let $j\in\{1,\ldots,2^\ell\}$.
				By assumption $d(\theta_j^{(\ell)},2^{-\ell}( \theta + j)) \leq 2\pi\eps_1$, 
				which using Claim~\ref{claim:sq-root} implies that $\lambda_j^{(\ell)} = e^{2i\pi\theta_j^{(\ell)}}$ has square roots $y_1=e^{2i\pi\varphi_1}$ and $y_2=e^{2i\pi\varphi_2}$ such that 
			\begin{equation}\label{eq:theta-a1}
			\max\big\{ d\big(\varphi_1, 	2^{-(\ell+1)}( \theta + j)\big),\, d\big(\varphi_2, 	2^{-(\ell+1)}( \theta + j + 2^\ell)\big)\big\} \,\leq\, \pi\eps_1.
			\end{equation}
				Since $\theta_j^{(\ell)}
        \in \Theta$, $\lambda_j^{(\ell)}$ is an eigenvalue
        of $U$, which using the first equation in~\eqref{eq:operator-bound} implies that both $y_1$ and $y_2$ are eigenvalues of $Y$.
        Using the second equation in~\eqref{eq:operator-bound} there must exist eigenvalues
        $\lambda_{2j-1}^{(\ell+1)}= e^{2i\pi \theta_{2j-1}^{(\ell+1)}}$ and
        $\lambda_{2j}^{(\ell+1)}= e^{2i\pi \theta_{2j}^{(\ell+1)}}$ of $Y^2$
				such that
        $$\max\big\{ \big|\lambda_{2j-1}^{(\ell+1)} - y_1
        \Big|,\, \big|\lambda_{2j}^{(\ell+1)} - y_2 \big|
        \big\}\,\leq\, \eps_1.$$
				Using Claim~\ref{claim:sq-root}, this implies 
				$$ \max\big\{ d( \theta_{2j-1}^{(\ell+1)},\varphi_1),\,d(\theta_{2j}^{(\ell+1)},\varphi_2)\big\}\,\leq\,\pi\eps_1.$$
	Using the triangle inequality and~\eqref{eq:theta-a1},
$$
\max\big\{  d\big(\theta_{2j}^{(\ell+1)}, 2^{-(\ell+1)}(\theta +j)\big),\,d\big(\theta_{2j-1}^{(\ell+1)} , 2^{-(\ell+1)}(\theta +j + 2^{\ell})\big) \big\} \leq \pi \eps_1 + \pi \eps_1, $$
        which (after relabeling the $\{\theta_j^{(\ell+1)},\,j\in\{1,\ldots,2^{\ell+1}\}\}$) completes the induction. 
    \end{proof}
    As long as $2\pi \eps_1 < 2^{-(\ell+1)}\pi$ the bound~\eqref{eq:ind-bound}
    implies that all $\theta_j^{(\ell)}$ are distinct modulo $1$, and hence
    that $d \geq 2^{\ell}$. If we choose $\ell$ so that $2^{-(\ell+2)} \leq 2
    \eps_1 \leq 2^{-(\ell+1)}$, then $d \geq 2^{\ell} > 1/(8\eps_1)$, proving
    the lemma. 
\end{proof}

Finally we prove the key lemma for the upper bound. 
\begin{lemma}\label{L:K-ubound}
    Let $d = 2^\ell$, where $\ell \geq 2$. Then there is a homomorphism
    \begin{equation*}
        \phi : \Free(\{a,b,c,x,y\}) \arr \mcU(\C^d)
    \end{equation*}
    such that $\phi(c) = -\Id$, and 
    \begin{equation*}
        \norm{\phi(r) - \Id}_{op} \leq 2^{-\ell}
    \end{equation*}
    for all defining relations $r$ of $K$.  
\end{lemma}
\begin{proof}
    Let $d=2^\ell$ be a power of two. Let $\{{f_i},\,i\in\{1,\ldots,d\}\}$ denote
    the canonical basis of $\C^d$. Let $$A = \begin{pmatrix} \Id & 0 \\ 0 & -\Id
    \end{pmatrix},\qquad B = \begin{pmatrix} -\Id& 0 \\ 0 & \Id \end{pmatrix}$$ in
    this basis, where each block has size $d/2$. Let $$Y = \sum_{j=1}^{d/2}
    \,{f_j}{f_{\frac{d}{2}+j}}^* +  \sum_{j=1}^{d/2}\, e^{2i\pi \theta_j}
    {f_{\frac{d}{2}+j}}{f_{j}}^*,$$ where for $j\in \{1,\ldots,d/2\}$, $\theta_j =
    j2^{-(\ell-1)}$. Then $Y$ is unitary and satisfies $YAY^\dagger = B$ and
    $YBY^\dagger = A$. For $j\in \{1,\ldots,d/2\}$, let $${u_j} =  \frac{{f_j}+
    e^{2i\pi \frac{\theta_j}{2}}{f_{\frac{d}{2}+j}}}{\sqrt{2}},\qquad  {v_j} =
    \frac{{f_j}-e^{2i\pi \frac{\theta_j}{2}}{f_{\frac{d}{2}+j}}}{\sqrt{2}}.$$ We
    can diagonalize $Y$ as $$ Y = \sum_{j=1}^{d/2} e^{2i\pi
    \frac{\theta_j}{2}}{u_j}u_j^* +  \sum_{j=1}^{d/2} e^{2i\pi
    \big(\frac{\theta_j}{2}+\frac{1}{2}\big)} {v_j}v_j^*, $$ and verify that $Y^2 =
    \sum_{j=1}^{d/2}  e^{2i\pi \theta_j} ({f_j}f_j^* +
    {f_{\frac{d}{2}+j}}f_{\frac{d}{2}+j}^*)$. Let $X$ be the unitary defined as
    follows:  if $j$ is even, $X{u_j}={f_{\frac{j}{2}}}$ and
    $X{v_j}={f_{\frac{d}{4}+\frac{j}{2}}}$, while if $j$ is odd, $X{u_j} =
    {f_{\frac{d}{2}+\frac{j+1}{2}}}$ and $X{v_j} =
    {f_{\frac{3d}{4}+\frac{j+1}{2}}}$. We verify that $X$ maps eigenvectors of $Y$
    to eigenvectors of $Y^2$ in a way such that the associated eigenvalues always
    differ by at most $2^{-\ell}$. 

    For the case of even $j$, the eigenvalue of $Y$ associated to $u_j$ is the same
    as the eigenvalue of $Y^2$ associated to $f_{\frac{j}{2}}$. The eigenvalue of
    $Y$ associated to $v_j$ is $e^{2i\pi(\frac{j2^{-(\ell-1)}}{2}+\frac{1}{2})}$,
    while the eigenvalue of $Y^2$ associated to $f_{\frac{d}{4}+\frac{j}{2}}$ is
    $e^{2i\pi (\frac{d}{4}+\frac{j}{2})2^{-(\ell-1)}}$, which given $d=2^\ell$ are
    identical.  

    For the case of odd $j$, the difference between the eigenvalue of $Y$
    associated to ${v_j}$ and the eigenvalue of $Y^2$ associated to
    $f_{\frac{3d}{4} + \frac{j+1}{2}}$ is
    \begin{align*}
        \Big|\Big(\frac{\theta_j}{2} + \frac{1}{2}\Big) -& \theta_{\frac{3d}{4}+\frac{j+1}{2}}\mod 1 \Big|\\
        &= \Big| \Big(j2^{-\ell} + \frac{1}{2} \Big)- \Big(\frac{j+1}{2} 2^{-(\ell-1)}+ \frac{3\cdot 2^\ell}{4}2^{-(\ell-1)}\Big)\mod 1 \Big|
        \\
        &= 2^{-\ell}.
    \end{align*}
    A similar calculation holds for the remaining case. Thus $\|XYX^\dagger - Y^2
    \| \leq 2^{-\ell}$. The proof of the lemma is concluded by defining a
    homomorphism $\phi:\Free(\{a,b,x,y\}) \arr \mcU(\C^{d})$  by $\phi(a)=A$,
    $\phi(b)=B$, $\phi(x)=X$ and $\phi(y)=Y$.
\end{proof}

\begin{proof}[Proof of Proposition \ref{P:bounds}]
Using Proposition~\ref{P:centralinv}, the lower bound is shown in Lemma~\ref{L:K-lbound}. The upper bound follows directly from Lemma~\ref{L:K-ubound} by noting that $\norm{\cdot}_f \leq \norm{\cdot}_{op}$. 
\end{proof}

\section{Linear system games}\label{S:linear}

We now turn to the connection between hyperlinear profile and linear system
games. Let $Ax=b$ be an $m \times n$ linear system over $\Z_2$. For every $1
\leq i \leq m$, let $V_i$ be set of indices of variables which appear in
equation $i$, i.e.  $V_i = \{ 1 \leq j \leq n : A_{ij} \neq 0 \}$. The linear
system non-local game associated to $Ax=b$ is played as follows \cite{CM14}:
\begin{enumerate}[(1)]
    \item The first player (Alice) is given an integer $i$, chosen uniformly at
        random from $\{1,\ldots,m\}$, and must reply with a vector ${y} \in
        \Z_2^{V_i}$ satisfying $\sum_{k \in V_i} y_k = b_i$. The output is
        interpreted as an assignment to the variables in equation $i$.
    \item The second player (Bob) is given an integer $j$, chosen uniformly at
        random from $\{1,\ldots,n\}$, and must reply with $x_j \in \Z_2$,
        interpreted as an assignment to the variable $x_j$. 
    \item The two players win if either $j \not\in V_i$, or $y_j = x_j$. 
\end{enumerate}
A quantum strategy $\mcS$ for the game (presented in terms of observables)
consists of
\begin{enumerate}[(a)]
    \item  two Hilbert spaces $H_A$ and $H_B$, along with a unit vector
        $\ket{\psi} \in H_A \otimes H_B$,
    \item a collection of $\{\pm 1\}$-valued observables $Y_{ij}$, $1 \leq i \leq m$,
        $j \in V_i$, on $H_A$ such that:
        \begin{enumerate}[(i)]
            \item $\prod_{j \in V_i} Y_{ij} = (-\Id)^{b_i}$ for all $1 \leq i \leq m$, and
            \item $Y_{ij} Y_{ij'} = Y_{ij'} Y_{ij}$ for all $1  \leq i \leq m$, $j,j' \in V_i$,
        \end{enumerate}
        and
    \item a collection of $\{\pm 1\}$-valued observables $X_j$, $1 \leq j \leq n$, on $H_B$.
\end{enumerate}
As mentioned in the introduction, a quantum strategy for a non-local game is
usually presented in terms of POVMs, rather than observables. A strategy as
above, presented in terms of observables, is equivalent to a strategy in the
usual form, in which the measurements are projective \cite{CM14,Sl17}. Indeed,
if $\{A_i^y\}_{y \in \Z_2^{V_i}}$, $\{B_j^c\}_{c\in \Z_2}$, $\ket{\psi}$ is a
strategy for a linear system game in which $\{A_i^y\}_y$ and $\{B_j^c\}_c$ are
projective measurements, then we can express this strategy in terms of
observables by setting $Y_{ij} = \sum_{{y}} (-1)^{y_j} A_i^{{y}}$ for each $1
\leq i \leq m$, $j\in V_i$, and $X_j = B_j^0- B_j^1$. In this way, the
observables $Y_{ij}$ and $X_j$ correspond to Alice's and Bob's assignments to
variable $x_j$. 

Using Naimark dilation, any strategy with POVMs can be turned into a strategy
with projective measurements. This multiplies the dimension of the Hilbert
space by the number of outputs.
\begin{lemma}\label{L:POVM}
    Let $\{A_i^a\}_{a=1}^k$ and $\{B_j^b\}_{b=1}^{\ell}$ be two families of
    POVMs, on Hilbert spaces $H_A$ and $H_B$ respectively, and let $\ket{\psi}$
    be a state in $H_A \otimes H_B$. Then there are families of
    projective measurements $\{\widehat{A}_i^a\}_{a=1}^k$ and
    $\{\widehat{B}_j^b\}_{b=1}^{\ell}$, on Hilbert spaces $\widehat{H}_A$
    and $\widehat{H}_B$ respectively, and a state $\ket{\widehat{\psi}} \in
    \widehat{H}_A \otimes \widehat{H}_B$, such that 
    \begin{equation*}
        \bra{\psi} A_i^a \otimes B_j^b \ket{\psi} = \bra{\widehat{\psi}} 
            \widehat{A}_i^a \otimes \widehat{B}_j^b \ket{\widehat{\psi}}  
    \end{equation*}
    for all $i,j,a,b$. Furthermore, the Hilbert spaces $\widehat{H}_A$ and
    $\widehat{H}_B$ can be chosen so that $\dim \widehat{H}_A = k \dim H_A$
    and $\dim \widehat{H}_B = \ell \dim H_B$. 
\end{lemma}
As a result, it suffices to prove Theorem \ref{T:main} for strategies with
projective measurements. Thus for the remainder of the paper, we assume
that all strategies are projective, and presented in terms of observables. 

Given a strategy $\mcS$ for a non-local game presented in terms of observables
as above, if $j \in V_i$, then Alice and Bob win with probability $p_{ij}(\mcS)
= (b_{ij}+1)/2$ on inputs $(i,j)$, where 
\begin{equation*}
    b_{ij}(\mcS) := \bra{\psi} Y_{ij} \otimes X_j \ket{\psi}.
\end{equation*}
The quantity $b_{ij}(\mcS)$ is called the \emph{success bias on inputs
$(i,j)$}. The strategy $\mcS$ is \emph{perfect} if $p_{ij}(\mcS)=1$ for all
inputs $(i,j)$, and \emph{$\eps$-perfect} if $p_{ij}(\mcS) \geq 1-\eps$ for all
inputs $(i,j)$.  Note that a strategy is perfect if and only if
$\omega(\mcG;\mcS)=1$, but it is not true that a strategy is $\eps$-perfect if
and only if $\omega(\mcG;\mcS)=1-\eps$.  It is not hard to show that the two
are related up to a factor depending on the size of the game:
\begin{lemma}\label{L:eps-perfect}
    Let $\mcG$ be the linear system game associated to the $m \times n$ linear
    system $Ax=b$. If $\mcS$ is an $\eps$-perfect strategy, then $\omega(\mcG;\mcS)
    \geq 1-\eps$. On the other hand, if $\mcS$ is a strategy such that $\omega(\mcG;\mcS)
    \geq 1-\eps$, then $\mcS$ is $(nm\eps)$-perfect. 
\end{lemma} 
\begin{proof}
    The game $\mcG$ has $nm$ pairs of inputs in total, each chosen with probability
    $1/(nm)$. Therefore success probability $1-\eps$ in the game implies that each
    question must lead to a valid answer with probability at least $1-nm\eps$. 
\end{proof}

The \emph{solution group} \cite{CLS16} associated to $Ax=b$ is the
finitely-presented group $\Gamma(A,b)$ generated by indeterminates
$x_1,\ldots,x_n$ and $J$, and satisfying the relations
\begin{enumerate}[(a)]
    \item $J^2 = e$, and $[x_j,J]=e$ for all $1 \leq j \leq n$, 
    \item $x_j^2 = e$ for all $1 \leq j \leq n$,
    \item $\prod_{j \in V_i} x_j = J^{b_i}$ for all $1 \leq i \leq m$, and
    \item $x_j x_{j'} = x_{j'} x_j$ for all $j,j' \in V_i$, $1 \leq i \leq m$.
\end{enumerate}
Let $\mcG$ be the linear system non-local game associated to the linear system
$Ax=b$. Several criteria relating representations of $\Gamma(A,b)$ to the
existence of perfect strategies for $\mcG$ have been given in
\cite{CM14,CLS16,Sl17}, including:
\begin{prop}[Proposition 3.4, \cite{Sl17}]\label{P:nontrivial}
    If $J$ is non-trivial in approximate representations of $\Gamma(A,b)$, then
    $\omega^q(\mcG) = 1$.
\end{prop}
For this paper, we need a quantitative version of Proposition
\ref{P:nontrivial}, and also its converse.  We start by giving the quantitative version, together with a converse for
maximally entangled states; although this is not necessary for the proof of
Theorem \ref{T:main}, it is easier to establish, and gives a tighter connection
with hyperlinear profile. The converse for general states will be given in the
next section.
\begin{prop}\label{P:approximate}
    Let $\mcG$ be the linear system non-local game associated to the $m \times
    n$ linear system $Ax = b$. 
    \begin{enumerate}[(a)]
        \item Given a $d$-dimensional $\eps$-representation $\phi$ of $\Gamma(A,b)$
            with $\phi(J) = -\Id$, we can construct an $O(\eps^2)$-perfect strategy
            for $\mcG$ which uses the maximally entangled state $\ket{\psi} \in
            \C^d \otimes \C^d$. 
        \item If $\ket{\psi},\{Y_{ij}\},\{X_j\}$ is an $\eps$-perfect strategy
            for $\mcG$, where $\ket{\psi} \in \C^d \otimes \C^d$ is maximally
            entangled, then the function 
            \begin{equation*}
                \phi : \Free(x_1,\ldots,x_n,J) \arr \mcU(\C^d)
            \end{equation*}
            sending $x_j \mapsto X_j$ and $J \mapsto -\Id$ is an
            $O(\sqrt{\eps})$-representation of $\Gamma(A,b)$.
    \end{enumerate}
\end{prop}
\begin{proof}
    For part (a), see the proof of Proposition 3.4 in \cite{Sl17} (the key
    idea is that $\Z_2^k$ is a stable group for any $k$). 

    For part (b), observe that if $\ket{\psi} \in \C^d \otimes \C^d$ is
    maximally entangled, we have that 
    \begin{equation*}
        \bra{\psi} A \otimes B \ket{\psi} = \ntr(A^T B),
    \end{equation*}
    where the transpose is taken with respect to the Schmidt basis for
    $\ket{\psi}$. Given an $\eps$-perfect strategy where $\ket{\psi} \in \C^d
    \otimes \C^d$ is maximally entangled, let $Z_{ij} = \overline{Y_{ij}}$,
    the entry-wise complex conjugate of $Y_{ij}$ with respect to the Schmidt
    basis of $\ket{\psi}$. Then 
    \begin{equation*}
        \ntr(Z_{ij}^* X_j) = \ntr(Y_{ij}^T X_j) = \bra{\psi} Y_{ij} \otimes X_j \ket{\psi} = b_{ij}(\mcS) \geq 1 - 2\eps.
    \end{equation*}
    Hence $\norm{Z_{ij} - X_j}_f^2 = 2 - 2 \ntr(Z_{ij}^* X_j) \leq O(\eps)$. For
    any $1 \leq i \leq m$, we have that
    \begin{equation*}
        \prod_{j \in V_i} Z_{ij} = (-\Id)^{b_i}, \text { and }
            Z_{ij} Z_{ij'} = Z_{ij'} Z_{ij} \text{ for all } j,j' \in V_i.
    \end{equation*}
    Hence 
    \begin{equation*}
        \Big\| \prod_{j \in V_i} X_j - (-\Id)^{b_i} \Big\|_f \leq O(\sqrt{\eps})
    \end{equation*}
    (where the constant depends on the size of $V_i$), and
    \begin{equation*}
        \norm{ X_{j} X_{j'} - X_{j'} X_{j}}_f \leq O(\sqrt{\eps}) \text{ for all }
            j,j' \in V_i.
    \end{equation*}
    Hence the homomorphism $\phi$ sending $x_j \mapsto X_j$ and $J \mapsto
    -\Id$ is an $O(\sqrt{\eps})$-representation of $\Gamma(A,b)$. 
\end{proof}
Proposition \ref{P:approximate} implies that, for strategies using maximally
entangled states, the Hilbert space dimension of near-perfect strategies is
given by the hyperlinear profile of the generator $J$ in the solution group.
To make this precise, we make the following definition.
\begin{defn}
    Given a non-local game $\mcG$, let $\omega^{me}(\mcG)$ be the supremum of
    winning probabilities $\omega(\mcG;\mcS)$ across all finite-dimensional
    quantum strategies $\mcS$ using a maximally entangled state. Let 
    $E^{me}(\mcG,\eps)$ be the smallest possible integer $d$ such that
    there is a quantum strategy $\mcS$ for $\mcG$ with 
    \begin{equation*}
        \omega(\mcG;\mcS) \geq \omega^{me}(\mcG) - \eps, 
    \end{equation*}
    where $\mcS$ uses a maximally entangled state in $\C^d \otimes \C^d$. 
\end{defn}

Lemma \ref{L:eps-perfect} and Proposition \ref{P:approximate} imply:
\begin{cor}\label{C:approximate}
    Let $\mcG$ be the linear system non-local game associated to a linear
    system $Ax = b$, and let $J$ be the generator $J \in \Gamma(A,b)$. If
    $\omega^{me}(\mcG) = 1$, then there are constants $C \geq C' > 0$ (depending on
    $\mcG$) such that
    \begin{equation*}
        \hlp(J,2,C\sqrt{\eps}) \leq E^{me}(\mcG,\eps) \leq \hlp(J,2, C' \sqrt{\eps})\;.
    \end{equation*}
\end{cor}

Note that, by Proposition \ref{P:centralinv}, the asymptotics of
$\hlp(J,2,\eps)$ are determined by $\hlp(J,\delta,\eps)$ for any $\delta \in
(0,2]$. 

\section{Near-optimal strategies with general states}\label{S:generalstates}

In this section we prove a version of part (b) of Proposition
\ref{P:approximate} for strategies with general states. Recall that any Hermitian matrix $X$ has a unique polar
decomposition $X = U \Sigma$, where $\Sigma$ is positive-semidefinite and $U$
is a unitary operator on $\Im \Sigma$. We write $|X|$ for the
positive-semidefinite part $\Sigma$, and $X |X|^{-1}$ for the unitary part $U$.

\begin{thm}\label{T:approximate2}
    Let $\mcG$ be the linear system non-local game associated to the $m \times
    n$ linear system $Ax = b$, and suppose $A$ has no non-zero columns. If
    $\ket{\psi},\{Y_{ij}\},\{X_j\}$ is an $\eps$-perfect strategy for $\mcG$,
    where $\ket{\psi} \in H_A \otimes H_B$, then there is a projection $P$ on
    $H_B$ such that
    \begin{equation*}
        \phi : \Free(x_1,\ldots,x_n,J) \arr \mcU(\Im P)
    \end{equation*}
    sending $x_j \mapsto (P X_j P) | P X_j P |^{-1}$ and $J \mapsto -P$ is an
    $O(\eps^{1/4})$-representation of $\Gamma(A,b)$ on $\Im P$.
\end{thm}
If the $j$th column of $A$ is zero, then the only relations containing $x_j$ in
$\Gamma(A,b)$ are $x_j^2 = e$ and $[x_j,J]=e$. So if we set $\phi(x_j)$ to an
arbitrary unitary when the $j$th row of $A$ is zero, then Theorem
\ref{T:approximate2} holds for any matrix $A$. 
Hence, the following corollary is an immediate consequence of Theorem
\ref{T:approximate2} and part (a) of Proposition \ref{P:approximate}:
\begin{cor}\label{C:approximate2}
    Let $\mcG$ be the linear system non-local game associated to a linear
    system $Ax=b$, and let $J$ be the generator $J \in \Gamma(A,b)$. If
    $\omega^q(\mcG)=1$, then there are constants $C \geq C' > 0$ (depending on
    $\mcG$) such that
    \begin{equation*}
        \hlp(J,2,C \eps^{1/4}) \leq E(\mcG,\eps) \leq \hlp(J,2,C' \sqrt{\eps}).
    \end{equation*} 
\end{cor}
In particular, if $\omega^q(\mcG)=1$ then $E(\mcG,\eps)$ is finite for all
$\eps > 0$, so we can strengthen Proposition \ref{P:nontrivial}:
\begin{cor}
    Let $\mcG$ be the linear system game associated to a linear system $Ax=b$.
    Then $\omega^q(\mcG)=1$ if and only if $J$ is non-trivial in approximate
    representations of $\Gamma(A,b)$.
\end{cor}

The proof of Theorem \ref{T:approximate2} has two ingredients. The first is a
state-dependent version of part (b) of Proposition \ref{P:approximate}.
\begin{prop}\label{P:generalstrat}
    Let $\mcG$ be the linear system non-local game associated to the $m \times
    n$ linear system $Ax=b$, and let $V_i = \{j : A_{ij} \neq 0\}$. Suppose
    that $A$ has no zero columns. If $\ket{\psi}$, $\{Y_{ij}\}$, $\{X_{j}\}$ is
    a finite-dimensional $\eps$-perfect strategy for $\mcG$, where $\ket{\psi}
    \in H_A \otimes H_B$, and $\rho$ is the reduced density matrix of
    $\ket{\psi}$ on $H_B$, then
    \begin{enumerate}[(a)]
        \item $\norm{X_j \rho^{1/2} - \rho^{1/2} X_j}_F \leq O(\sqrt{\eps})$ for all $1 \leq j \leq n$
,
        \item $\norm{\prod_{j \in V_i} X_j - (-\Id)^{b_i}}_\rho \leq O(\sqrt{\eps})$
            for all $1 \leq i \leq m$, where the constant depends on the size of $V_i$,
            and
        \item $\norm{X_j X_k - X_k X_j}_{\rho} \leq O(\sqrt{\eps})$ for all
            $1 \leq i \leq m$, $j,k \in V_i$.
    \end{enumerate}
\end{prop}
\begin{proof}
    Without loss of generality, we can pick a basis $\{\ket{t}\}$ of $H_A$ such
    that $\ket{\psi} = \sum_t \ket{t} \lambda \ket{t}$, where $\lambda =
    \rho^{1/2}$. For any matrices $B$ and $C$,
    \begin{equation*}
        \bra{\psi} B^T \otimes C \ket{\psi} = \tr(B \lambda C \lambda), 
    \end{equation*}
    where the transpose is taken with respect to the basis $\{\ket{i}\}$.
    Suppose that $B$ and $C$ are unitary, and that
    \begin{equation*}
        \bra{\psi} B^T \otimes C \ket{\psi} \geq 1-O(\eps).
    \end{equation*}
    We claim that
    \begin{equation}\label{E:Ccommuting}
        \norm{C\lambda - \lambda C}_F \leq O(\sqrt{\eps}), 
    \end{equation}
    and that
    \begin{equation}\label{E:CtoB}
        \norm{C\lambda - \lambda B^*}_F \leq O(\sqrt{\eps}).
    \end{equation}
    To prove equation \eqref{E:Ccommuting}, observe that $\tr(B \lambda C
    \lambda)$ can be thought of as the Frobenius inner product of
    $\lambda^{1/2} B^* \lambda^{1/2}$ and $\lambda^{1/2} C \lambda^{1/2}$.
    Applying the Cauchy-Schwarz inequality, we get that
    \begin{equation*}
        1 - O(\eps) \leq \tr(B \lambda C \lambda) \leq 
            \norm{\lambda^{1/2} B^* \lambda^{1/2}}_F
                \norm{\lambda^{1/2} C \lambda^{1/2}}_F \leq \norm{\lambda^{1/2} C \lambda^{1/2}}_F,
    \end{equation*}
    where the last inequality uses the fact that
    \begin{equation*}
        \norm{\lambda^{1/2} B^* \lambda^{1/2}}_F^2 = \tr(B \lambda B^* \lambda)
            \leq \norm{B \lambda B^*}_F \norm{\lambda}_F = 1.
    \end{equation*}
    Since we can assume that $\eps \leq 1$ (so that $\eps^2 = O(\eps)$), we get that
    \begin{equation*}
        \norm{C \lambda - \lambda C}_F^2 = 2 - 2 \tr(C^* \lambda C \lambda)
            = 2 - 2 \norm{\lambda^{1/2} C \lambda^{1/2}}_F^2 \leq O(\eps).
    \end{equation*}
    Equation \eqref{E:CtoB} follows immediately from the fact that $\lambda
    B^*$ and $C \lambda$ are unit vectors in the Frobenius norm.

    Now for part (a), if $j \in V_i$ then by assumption
    \begin{equation}\label{eq:xy-1}
        \bra{\psi} Y_{ij} \otimes X_j \ket{\psi} \geq 1 - O(\eps). 
    \end{equation}
    Since $A$ has no zero columns, every $j$ appears in some $V_i$, so we conclude
    from equation \eqref{E:Ccommuting} that
    \begin{equation*}
        \norm{X_j \lambda - \lambda X_j}_F \leq O(\sqrt{\eps}).
    \end{equation*}

    For part (b), fix $1 \leq i \leq m$. Using equations~\eqref{eq:xy-1} and
    \eqref{E:CtoB}, as well as the unitary invariance of the Frobenius norm, we
    get that
    \begin{equation*}
        \Big\|\prod_{j \in V_i} X_j \lambda - \lambda \prod_{j \in V_i} \overline{Y_{ij}}\Big\|_F \leq O(\sqrt{\eps}),
    \end{equation*}
    where the constant depends on the size of $V_i$. Since $\prod_{j \in V_i}
    \overline{Y_{ij}} = (-\Id)^{b_i}$, we get that
    \begin{equation*}
        \Big\|\prod_{j \in V_i} X_j - (-\Id)^{b_i}\Big\|_\rho = \Big\|\Big(\prod_{j \in
            V_i} X_j - (-\Id)^{b_i}\Big) \lambda\Big\|_F \leq O(\sqrt{\eps}).
    \end{equation*}
    The same argument, using the fact that $Y_{ij}$ and $Y_{ik}$ commute, can be used for part (c).
\end{proof}

Let $\chi_I$ denote the indicator function of a subset $I\subseteq \R$, and for $a\in\R$ let
$\chi_{\geq a} := \chi_{[a,+\infty)}$. Note that if $\lambda$ is a self-adjoint
operator (and $I$ is measurable), then $\chi_I(\lambda)$ is a projection.  The
second ingredient in the proof of Theorem \ref{T:approximate2} is the following
version of Connes' joint distribution trick:
\begin{lemma}[\cite{Co76}, Lemma 1.2.6]\label{L:connestrick}
    Let $\lambda, \lambda'$ be positive semidefinite operators on a finite-dimensional
    Hilbert space. Then 
    \begin{equation*}
        \int_0^{+\infty} \norm{\chi_{\geq \sqrt{a}}(\lambda) 
            - \chi_{\geq \sqrt{a}}(\lambda')}^2_F da 
            \leq \norm{\lambda - \lambda'}_F \norm{\lambda + \lambda'}_F.
    \end{equation*}
\end{lemma}
Connes' argument applies more generally to any semifinite von Neumann algebra
with a normal semifinite faithful trace.  For
the convenience of the reader, we give a self-contained proof (following the
original) for the finite-dimensional case. 
\begin{proof}
    The proof is based on the following trick due to Connes: for any positive
    semidefinite operators $\lambda, \lambda'$ on a finite-dimensional Hilbert
    space, there is a discrete measure $\nu$ on $\R_{\geq 0} \times
    \R_{\geq 0}$ such that
    \begin{equation}\label{E:connes-trick}
        \norm{f(\lambda)-g(\lambda')}_F^2 = \int_{(x,y)} |f(x)-g(y)|^2 \; d\nu
    \end{equation}
    for any pair of functions $f,g:\R_+\to\R_+$. Indeed, if we write the spectral decomposition $\lambda =
    \sum_i \lambda_i \ket{u_i}\!\bra{u_i}$, $\lambda' = \sum_j \lambda'_i
    \ket{v_j}\!\bra{v_j}$, and set $\nu = \sum_{i,j} \delta_{(\lambda_i,\mu_j)}
    |\bra{u_i}v_j\rangle|^2$, then Equation \eqref{E:connes-trick} follows
    by direct calculation. 

    Now for non-negative real numbers $x,y$,  
    \begin{align}
        \int_0^{+\infty} \big|\chi_{\geq \sqrt{a}}(x)-\chi_{\geq \sqrt{a}}(y)\big|^2 da
        &= |x^2-y^2| = |x-y||x+y|\;.\label{E:connes-trick-2}
    \end{align}
    Thus, using~\eqref{E:connes-trick} followed by Fubini's theorem,
    \begin{align*}
        \int_0^{+\infty} \norm{\chi_{\geq \sqrt{a}}(\lambda) - \chi_{\geq \sqrt{a}}(\lambda')}^2_F da
        &=  \int_{(x,y)} \int_0^{+\infty} \big|\chi_{\geq \sqrt{a}}(x)-\chi_{\geq \sqrt{a}}(y)\big|^2  da\,d\nu \\
        &=  \int_{(x,y)} |x-y||x+y| d\nu \\
        &\leq  \left(\int_{(x,y)} |x-y|^2d\nu \right)^{1/2}\left(\int_{(x,y)} |x+y|^2d\nu \right)^{1/2} \\
        &= \|\lambda-\lambda'\|_F \|\lambda+\lambda'\|_F\;,
    \end{align*}
		where the last equality again uses~\eqref{E:connes-trick}.
\end{proof}

We need one other easy lemma:
\begin{lemma}\label{L:int}
    Let $\lambda$ be a compact positive semidefinite operator on a Hilbert space. Then
    \begin{equation*}
        \int_{0}^{+\infty} \chi_{\geq \sqrt{a}}(\lambda) da = \lambda^2.
    \end{equation*}
\end{lemma}
\begin{proof}
    It suffices to prove the lemma for $\lambda = t \ket{v}\bra{v}$, where
    $\ket{v}$ is a unit vector, and $t \geq 0$. But then
    \begin{equation*}
        \chi_{\geq \sqrt{a}}\left(\lambda\right) = 
            \begin{cases} \ket{v} \bra{v} & a \leq t^2 \\
                            0 & a > t^2 \end{cases},
    \end{equation*}
    so
    \begin{equation*}
        \int_0^{+\infty} \chi_{\sqrt{a}}(\lambda) da = 
            \int_0^{t^2} \ket{v} \bra{v} da = t^2 \ket{v} \bra{v} = \lambda^2.
    \end{equation*}
\end{proof}

\begin{proof}[Proof of Theorem \ref{T:approximate2}]
    Let $V_i = \{j : A_{ij} \neq 0\}$, and suppose $\ket{\psi}$, $\{Y_{ij}\}$,
    $\{X_j\}$ is an $\eps$-perfect strategy for $\mcG$. Let $\rho$ be the
    reduced density matrix of $\ket{\psi}$ on $H_B$, and let $\lambda =
    \rho^{1/2}$. By Lemma \ref{L:int}, for any Hermitian $W$ on $H_B$,
    \begin{equation}\label{E:normalization}
        \int_{0}^{+\infty} \tr(W\chi_{\geq \sqrt{a}}(\lambda)) da = \tr(W\rho) \;.
    \end{equation}
    Then, if $R$ is one of the relations $\prod_{j \in V_i} X_j - (-\Id)^{b_i}$, $1 \leq i \leq m$,
    or $X_j X_k - X_k X_j$, $j,k \in V_i$,
    \begin{equation}\label{E:relations}
        \int_0^{+\infty} \norm{ R \,\chi_{\geq \sqrt{a}}(\lambda)}_F^2 da = 
            \norm{ R}_{\rho}^2 \leq O(\eps),
    \end{equation}
    where the equality is by~\eqref{E:normalization}, using that $\chi_{\geq \sqrt{a}}(\lambda)$ is a projection, and the inequality follows from parts (b) and (c) of Proposition
    \ref{P:generalstrat}. Finally,  Lemma \ref{L:connestrick} and part (a) of Proposition \ref{P:generalstrat} imply that
    \begin{multline}\label{E:commute}
        \int_0^{+\infty} \norm{\chi_{\geq \sqrt{a}}(\lambda) - X_j^* \chi_{\geq \sqrt{a}}(\lambda) X_j}_F^2 da  \\
            \leq \norm{\lambda - X_j^* \chi_{\geq \sqrt{a}}(\lambda) X_j}_F \norm{\lambda + X_j^* \chi_{\geq \sqrt{a}}(\lambda)
                X_j}_F \leq O(\sqrt{\eps})
    \end{multline}
    for all $1 \leq j \leq n$.  Putting equations~\eqref{E:normalization},
    \eqref{E:relations}, and \eqref{E:commute} together, and using the fact
    that $\eps \leq \eps^{2}$ for $\eps$ small enough,
    \begin{align*}
      &  \int_0^{+\infty}  \sum_{j=1}^n \norm{\chi_{\geq \sqrt{a}}(\lambda) - X_j^* \chi_{\geq \sqrt{a}}(\lambda) X_j}_F^2
            + \sum_i \Big\|\Big( \prod_{j \in V_i} X_j - (-\Id)^{b_i}\Big) \chi_{\geq \sqrt{a}}(\lambda) \Big\|_F^2 \\
            &\; + \sum_i \sum_{j \neq k \in V_i} \norm{ \left( X_j X_k - X_k X_j \right) \chi_{\geq \sqrt{a}}(\lambda)}_F^2\ \  da 
                \leq O(\eps^{1/2}) \int_{0}^{+\infty} \tr(\chi_{\geq \sqrt{a}}(\lambda)) da.
    \end{align*}
    All the integrands in the above equation are zero if $a > \norm{\lambda}_{op}^2$, so it follows that we can
    find $0 \leq a_0 \leq \norm{\lambda}_{op}^2$ such that if $P := \chi_{\geq \sqrt{a_0}} (\lambda)$, then
    \begin{multline}
        \sum_{j=1}^n \norm{P - X_j^* P X}_F^2
            + \sum_i \Big\|\Big( \prod_{j \in V_i} X_j - (-\Id)^{b_i}\Big) P \Big\|_F^2 \\
             + \sum_i \sum_{j \neq k \in V_i} \norm{ \left( X_j X_k - X_k X_j \right) P }_F^2
                \leq O(\eps^{1/2}) \tr(P).
    \end{multline}
    Since all summands on the left are positive, we conclude that all summands
    are bounded by $O(\eps^{1/2}) \tr(P)$ (where the constant depends only on
    $m$ and $n$). In addition, $a \leq \norm{\lambda}_{op}^2$ implies that $P
    \neq 0$. Let $\widetilde{X}_j := (P X_j P) | P X_j P|^{-1}$ for
    $j=1,\ldots,n$. We want to show that the homomorphism $\phi$ from $\Free(S)$ to
    $\mcU(\Im P)$ sending $x_j \mapsto \widetilde{X}_j$ is an
    $O(\eps^{1/4})$-representation.  To start, we need to show that $P X_j P$
    is almost unitary, so $\widetilde{X}_j$ is close to $P X_j P$. Note that
    if $B B^* \leq \Id$ (or equivalently, $B^* B \leq \Id$) then $\norm{A B}_F \leq
    \norm{A}_F$ and similarly $\norm{BA}_F \leq \norm{A}_F$ for any matrix $A$.
    As a result,
    \begin{equation*}
        \norm{X_j P - P X_j P}_F = \norm{(X_j P - P X_j) P}_F \leq \norm{X_j P - P X_j}_F \leq O(\eps^{1/4}) \tr(P)^{1/2},
    \end{equation*}
    and hence
    \begin{align*}
        \norm{ (P X_j P)^2 - P}_F & = \norm{ P X_j (P X_j P - X_j P)}_F \\ 
                &\leq \norm{P X_j P - X_j P}_F \\
                &\leq O(\eps^{1/4}) \tr(P)^{1/2}.
    \end{align*}
    By Lemma \ref{L:stability1}, part (a),
    \begin{equation*}
        \big\|\widetilde{X}_j - P X_j P\big\|_F \leq O(\eps^{1/4}) \tr(P)^{1/2},
    \end{equation*}
    and we also have
    \begin{equation*}
        \big\|\widetilde{X}_j - X_j P\big\|_F \leq O(\eps^{1/4}) \tr(P)^{1/2}.
    \end{equation*}

    To finish the proof, we claim that
    \begin{equation*}
        \big\|X_{i_1} \cdots X_{i_k} P - \widetilde{X}_{i_1} \cdots \widetilde{X}_{i_k}\big\|_F \leq O(\eps^{1/4}) \tr(P)^{1/2}
    \end{equation*}
    for any $1 \leq i_1,\ldots,i_k \leq n$, where the constant depends on $k$.
    Indeed, for any $1 \leq j \leq n$, if $B = \widetilde{X}_{i_{j+1}} \cdots \widetilde{X}_{i_{k}} P$
    ($P$ is included at the end for the case $j=n$)
    then $B^* B = P \leq \Id$, thus $BB^*\leq \Id$ and 
    \begin{align*}
        \big\|X_{i_1} \cdots X_{i_j} \widetilde{X}_{i_{j+1}} \cdots \widetilde{X}_{i_k} P 
            - X_{i_1}&\cdots X_{i_{j-1}} \widetilde{X}_{i_j} \cdots  \widetilde{X}_{i_k}P\big\|_F \\
                        & = \big\| \big(X_{i_j} P - \widetilde{X}_{i_j} \big) \widetilde{X}_{i_{j+1}} \cdots \widetilde{X}_{i_k} P 
            \big\|_F\\
&            \leq \big\|X_{i_j} P - \widetilde{X}_{i_j}\big\|_F \leq O(\eps^{1/4}) \tr(P)^{1/2}\;,
    \end{align*}
    so the claim follows. Thus if $\phi' : \mcF(S) \arr \mcU(H_B)$ is defined by sending
    $x_j \mapsto X_j$ and $J \mapsto -\Id$, then 
    \begin{align*}
        \norm{\phi(r) - P}_f &=  \frac{1}{\tr(P)^{1/2}}    \norm{\phi(r) - P}_F\\
&	\leq \frac{1}{\tr(P)^{1/2}} \left(\norm{\phi(r) - \phi'(r)P}_F + \norm{(\phi'(r)- \Id) P}_F \right) \\
&\leq O(\eps^{1/4}).
    \end{align*}
    for any of the defining relations $r$ of $\Gamma(A,b)$.
\end{proof}

\begin{rmk}\label{rk:connes}
    When combined with Proposition \ref{P:approximate}, part (a), Theorem
    \ref{T:approximate2} implies that any $\eps$-perfect strategy can be turned
    into an $O(\eps^{1/2})$-strategy with a maximally entangled state.  This
    can also be proved directly using Lemma \ref{L:connestrick}, and when
    combined with part (b) of Proposition \ref{P:approximate}, this gives
    another route to prove Theorem \ref{T:approximate2}. The Connes joint
    distribution trick can also be used to prove similar statements about other
    classes of games, such as synchronous games~\cite{paulsen2016estimating} and weak projection games~\cite{manvcinska2014maximally}. 
    Determining the largest class of games to which this idea can be applied seems to be an interesting question for further work. \end{rmk}

\section{Bounds on hyperlinear profile of a specific solution group}\label{S:embedding}

To apply Corollary~\ref{C:approximate2}, we need examples of solution groups
with bounds on the hyperlinear profile of $J$. In this section, we use
Proposition \ref{P:bounds} to prove:
\begin{prop}\label{P:sgbounds}
    There is a solution group $\Gamma(A,b)$ and constants $C,C' > 0$ such
    that $J$ is non-trivial in approximate representations of $\Gamma(A,b)$,
    and furthermore
    \begin{equation*}
        \frac{C}{\eps^{2/3}} \leq \hlp(J,2,\eps) \leq \frac{C'}{\eps}
    \end{equation*}
    for all $\eps > 0$.
\end{prop}
To prove Proposition \ref{P:sgbounds}, recall from \cite{Sl17} that a group is
an \emph{extended homogeneous-linear-plus-conjugacy} group if it has a
presentation $\langle S : R \rangle$, where 
\begin{enumerate}[(a)]
    \item the set $S$ of generators is split into two disjoint subsets $S_0$
        and $S_1$, where $S_1 = \{y_1,\ldots y_{\ell}\}$ is linearly-ordered,
    \item $R$ contains the relations $x^2 = e$ for all $x \in S_0$, 
    \item there is a collection $\mcV$ of ordered subsets of $S_0$ such that
        for all $V \in \mcV$, $R$ contains the relations
        \begin{equation*}
            \prod_{x \in V} x = e, \text{ and } xy = yx \text{ for all distinct }
                x,y \in V 
        \end{equation*}
        (these relations are similar to the relations of a solution group), and
    \item every other relation of $R$ is of the form
        \begin{enumerate}[(i)]
            \item $x y x = z$ for some (not necessarily distinct) $x,y,z \in S_0$,
            \item $x y x^{-1} = z$ for some $x \in S_1$ and (not necessarily 
                distinct) $y,z \in S_0$, or
            \item $y_i y_j y_{i}^{-1} = y_j^{k}$ for some $1 \leq j < i \leq
                n$ and $k>0$.  
        \end{enumerate}
\end{enumerate}
We call a presentation $G = \langle S : R \rangle$ of this form a
\emph{presentation of $G$ as an extended homogeneous-linear-plus-conjugacy
group}. We can use the main argument of \cite{Sl17} to prove:
\begin{prop}\label{P:embedding}
    Let $G = \langle S : R \rangle$ be a presentation of $G$ as an extended
    homogeneous-linear-plus-conjugacy group, and suppose $w \in S_0$ represents a
    central involution in $G$ which is non-trivial in approximate
    representations.  Then there is a solution group $\Gamma(A,b)$ and
    constants $ C > C' > 0$, $N \in \mbN$, such that $J$ is non-trivial in
    approximate representations of $\Gamma(A,b)$, and furthermore
    \begin{equation*}
        \hlp(w,2,C\eps) \leq \hlp(J,2,\eps) \leq N \hlp(w,2,C'\eps)
    \end{equation*}
    for all $\eps > 0$.
\end{prop}
We need one lemma before proving Proposition \ref{P:embedding}.
\begin{lemma}\label{L:anticommutator}
    If $A$ and $B$ are unitary matrices, then
    \begin{equation*}
        \norm{\Id - A}_f \geq 1 - \frac{1}{2} \norm{\Id + [A,B]}_f.
    \end{equation*}
\end{lemma}
\begin{proof}
    Observe that $\norm{\Id - [A,B]}_f \leq 2 \norm{\Id - A}_f$. 
    Hence
    \begin{equation*}
        2 = \norm{2 \Id}_f \leq \norm{\Id - [A,B]}_f + \norm{\Id + [A,B]}_f
            \leq 2 \norm{\Id - A}_f + \norm{\Id + [A,B]}_f.
    \end{equation*}
\end{proof}
\begin{proof}[Proof of Proposition \ref{P:embedding}]
    We combine several propositions from \cite{Sl17}. By \cite[Proposition
    4.8]{Sl17}, $G$ embeds in a homogeneous-linear-plus-conjugacy group $G'$,
    and the embedding can be chosen so that the image of $w$ is non-trivial in
    approximate representations of $G'$. We then make $G'$ into a
    linear-plus-conjugacy group $\widehat{G}$  by adding two new generators $t$
    and $J$, along with relations stating that $J$ is central, $J^2 = t^2 = e$,
    and $t w t = J$.  Since $\widehat{G}$ is an HNN extension of $G' \times
    \Z_2$, $G'$ embeds in $\widehat{G}$, and \cite[Lemma 5.2]{Sl17} implies
    that $J$ is non-trivial in approximate representations of $\widehat{G}$.
    Finally, \cite[Proposition 4.2]{Sl17} implies that $\widehat{G}$ embeds in
    a solution group $\Gamma(A,b)$, where $J \in \widehat{G}$ maps to $J \in
    \Gamma(A,b)$, and $J$ is non-trivial in approximate representations of
    $\Gamma(A,b)$. Suppose $A$ is an $m \times n$ matrix. We conclude that
    there is a homomorphism $\Psi : \Free(S) \arr \Free(x_1,\ldots,x_n,J)$
    descending to an embedding $G \incl \Gamma(A,b)$, such that if $z$ is the
    image of $t$ in $\Free(x_1,\ldots,x_n,J)$, then $[\Psi(w),z] = J$ in
    $\Gamma(A,b)$.  

    To show the lower bound on $\hlp(J,2,\eps)$, note that there is a
    constant $C_0 \geq 1$ such that if $\phi$ is an $\eps$-representation
    of $\Gamma(A,b)$ with $\phi(J)=-\Id$, then
    \begin{equation*}
        \norm{[\phi(\Psi(w)),\phi(z)] + \Id}_f \leq C_0 \eps.
    \end{equation*}
    By Lemma \ref{L:anticommutator}, $\norm{\phi(\Psi(w)) - \Id}_f \geq 1 -
    O(\eps)$. We conclude that 
    \begin{equation*}
        \hlp(w,1-O(\eps),O(\eps)) \leq \hlp(\Psi(w),1-O(\eps),\eps) \leq \hlp(J,2,\eps),
    \end{equation*}
    where the first inequality follows from Lemma \ref{L:hlpprop2}. Finally,
    since $w$ represents a central involution, we conclude from Proposition
    \ref{P:centralinv} that $\hlp(w,2,O(\eps)) \leq \hlp(J,2,\eps)$.
    
    The upper bound on $\hlp(J,2,\eps)$ requires careful attention to the
    construction of $\Psi$ described above. Suppose $\phi$ is a $d$-dimensional
    $\eps$-representation of $G$ with $\phi(w) = -\Id$. Then the proof of
    \cite[Proposition 4.8]{Sl17} (including the change of presentation) implies
    that there is a $2d$-dimensional $\eps$-representation $\gamma$ of $G'$
    with $\gamma(w) = -\Id_d \oplus \Id_d$. Setting
    \begin{equation*}
        \gamma(J) = -\Id \text{ and } \gamma(t) = \begin{pmatrix} 0 & \Id \\ \Id & 0
                                                \end{pmatrix},
    \end{equation*}
    we get a $2d$-dimensional $\eps$-representation $\gamma$ of $\widehat{G}$
    with $\gamma(J) = -\Id$. Finally \cite[Remark 4.5]{Sl17} implies that there
    is an $8d$-dimensional $O(\eps)$-representation $\gamma'$ of $\Gamma(A,b)$
    with $\gamma'(J) = -\Id$. We conclude that $\hlp(J,2,O(\eps)) \leq 8
    \hlp(w,2,\eps)$, so the proposition follows with $N=8$. 
\end{proof}

\begin{proof}[Proof of Proposition \ref{P:sgbounds}]
    Consider the group $K$ defined in Section \ref{S:bounds}. Since $abc=e$ and
    $a$, $b$, and $c$ are involutions, this relation implies that the elements
    $a$, $b$, and $c$ pairwise commute. Hence $K$ is an extended
    homogeneous-linear-plus-conjugacy group. Changing the presentation by adding
    the relations $ab = ba$, $ac = ca$, and $bc = cb$ does not change the
    asymptotics of the hyperlinear profile of $c$, so by Proposition
    \ref{P:embedding} there is a solution group $\Gamma(A,b)$ with 
    \begin{equation*}
        \Omega(1/\eps^{2/3}) \leq \hlp(J,2,\eps) \leq O(1/\eps).
    \end{equation*}
\end{proof}

\begin{proof}[Proof of Theorem \ref{T:main}]
    Upper and lower bounds both follow directly from Corollary \ref{C:approximate2}
    and Proposition \ref{P:sgbounds}.
\end{proof}

\section{A presentation-independent version of hyperlinear profile}\label{S:presindep}

A weakness of Definition \ref{D:hlp1} is that it depends on the
presentation of $G$. Following \cite{Co13}, we can define another version of
hyperlinear profile which is independent of the presentation, while
still being in the spirit of Definition \ref{D:hlp1}. 
\begin{defn}\label{D:hlp2}
    Let $E$ be a finite subset of a group $G$ containing the identity. 
    Let $\eta(E)$ be the function $\R_{>0} \times \R_{>0} \arr \mbN \cup
    \{+\infty\}$ such that $\eta(E;\delta,\eps)$ is the smallest
    integer $d$ for which there is a function $\phi : E \arr \mcU(\C^d)$
    with  
    \begin{enumerate}[(a)]
        \item $\phi(e) = \Id$, 
        \item $\norm{\phi(x)\phi(y) - \phi(xy)}_f \leq \eps$ whenever $x$, $y$,
            and $xy$ are in $E$, and
        \item $\norm{\phi(x) - \phi(y)}_f \geq \delta$ whenever $x \neq y$
            in $E$, 
    \end{enumerate}
    or $+\infty$ if no such function $\phi$ exists. 
\end{defn}
Given a finite subset $E$ of $G$ containing the identity, let
$\sigma(E;\delta,\eps)$ be defined similarly to $\eta(E;\delta,\eps)$, but with
$\mcU(\C^d)$ replaced by the group $S_d$ of $d \times d$ permutation matrices.
Then $\sigma(E)$ is a two-parameter version of the sofic profile defined in
\cite{Co13}. Specifically, the \emph{sofic profile} of $E$ is 
\begin{equation}\label{E:sofic}
    \mbN \arr \mbN \cup \{+\infty\} : n \mapsto \sigma\left(E; \sqrt{2 - \frac{2}{n}}, \sqrt{\frac{2}{n}}\right).
\end{equation}
More precisely, the sofic profile of $E$ is the equivalence class of this
function under the equivalence relation which identifies non-decreasing
functions with the same asymptotic growth rate. However, for our purposes we
can ignore this distinction. The square roots in Equation \eqref{E:sofic} come
from the fact that sofic profile is defined in terms of the normalized Hamming metric on
$S_n$, which for $u,v \in S_n$ is $\norm{u - v}_f^2 / 2$. 

Let $G = \langle S : R \rangle$ be a finitely-presented group. Then the
collection of functions $\{\eta(E)\}$ and $\{\hlp(T)\}$ are roughly equivalent,
in the sense that any function in one collection is bounded (asymptotically) by
some function in the other collection. 
\begin{prop}\label{P:roughequiv}
    Let $G = \langle S : R \rangle$ be a finitely-presented group. 
    \begin{enumerate}[(a)]
        \item Suppose $T \subset \Free(S)$ is a finite subset not containing
            any element which is trivial in $G$. Then there is a constant $C > 0$
            and a finite subset $E$ of $G$ containing the identity such that for all $\eps,\delta>0$,
            \begin{equation*}
                \hlp(T, \delta - C \eps, C \eps) \leq \eta(E,\delta,\eps).
            \end{equation*}
        \item Suppose $E \subset G$ is a finite subset containing the identity.
            Then there is a constant $C > 0$ and a finite subset $T \subset \Free(S)$
            not containing any element which is trivial in $G$, such that for all $\eps,\delta>0$,
            \begin{equation*}
                \eta(E,\delta,C \eps) \leq \hlp(T,\delta,\eps).
            \end{equation*}
    \end{enumerate}
\end{prop}
\begin{proof}
    For (a), suppose $T \subset \Free(S)$.  Let $E \subset G$ be the set of all
    elements of the form $x_1^{a_1} x_2^{a_2} \cdots x_i^{a_i}$, where
    $x_1^{a_1} \cdots x_{n}^{a_n}$ is an element of $T \cup R \cup S$ for some
    (not necessarily distinct) elements $x_1,\ldots,x_n \in S$, integers
    $a_1,\ldots,a_n \in \Z$, and $0 \leq i \leq n$. Note that $E$ contains the
    identity. Let $C$ be the length of the longest word in $T \cup R$, and
    suppose $\phi : E \arr \mcU(\C^d)$ has the property that
    $\norm{\phi(x)\phi(y) - \phi(xy)}_f \leq \eps$ whenever $x$, $y$, and $xy$
    are in $E$. Define $\psi : \Free(S) \arr \mcU(\C^d)$ by $\psi(s) = \phi(s)$
    for all $s \in S$. By the definition of $E$, if $w \in T \cup R$ then
    $\norm{\psi(w) - \phi(w)}_f \leq C \eps$. Thus $\norm{\psi(r) - \Id}_f \leq
    C \eps$ for all $r \in R$, and $\norm{\psi(w) - \Id}_f \geq \norm{\phi(w) -
    \Id}_f - C \eps$ for all $w \in T$. We conclude that $\hlp(T,\delta-C\eps,C\eps)
    \leq \eta(E,\delta,\eps)$ for all $\delta,\eps > 0$.

    For part (b), suppose $E \subset G$ is finite. Choose a representative
    $w(u) \in \Free(S)$ for every $u \in E \setminus \{e\}$, and set $w(e_G) = e
    \in \Free(S)$. Let $T \subset \Free(S)$ to be the set of elements of the
    form $w(u) w(v)^{-1}$, where $u$ and $v$ are distinct elements of $E$, and
    choose $C$ such that for any triple $u$, $v$, and $uv$ in $E$, the word
    $w(u) w(v)$ can be turned into $w(uv)$ by at most $C$ applications of the
    relations in $R$. Suppose $\psi : \Free(S) \arr \mcU(\C^d)$ is an
    $\eps$-representation of $G$, and define $\phi : E \arr \mcU(\C^d)$ by
    $\phi(u) = \psi(w(u))$. Note that $\phi(e) = \Id$. Then
    \begin{equation*}
        \norm{\phi(u)\phi(v) - \phi(uv)}_f = \norm{\psi(w(u)w(v)) - \psi(w(uv))}_f
            \leq C\eps
    \end{equation*}
    for all triples $u$, $v$, and $uv$ in $E$. If $u$ and $v$ are distinct
    elements of $E$, then $\norm{\phi(u) - \phi(v)}_f =
    \norm{\psi(w(u)w(v)^{-1}) - \Id}_f$, so we conclude that
    $\eta(E,\delta,C\eps) \leq \hlp(T,\delta,\eps)$ for all $\delta,\eps > 0$. 
\end{proof}
We do not know whether there is a finer way to regard the families
$\{\eta(E)\}$ and $\{\hlp(T)\}$ as asymptotically equivalent, aside from the
rough comparison in Proposition \ref{P:roughequiv}. This raises the question of
whether we should call $\{\eta(E)\}$ the hyperlinear profile of $G$, and find a
different term for $\{\hlp(T)\}$. However, the point of hyperlinear profile
(and this is also true of sofic profile) is to see how fast the functions
$\eta(E)$ or $\hlp(T)$ can grow. Introducing another term would also raise the
possibility of confusion with another related concept, the sofic dimension
growth of Arzhantseva and Cherix (see \cite{Ca16}), which is in a somewhat
different spirit. Thus we suggest that both families $\{\eta(E)\}$ and
$\{\hlp(T)\}$ should be regarded as the hyperlinear profile of $G$, with a
specific definition selected to fit the context. 

\bibliographystyle{amsalpha}
\bibliography{approx}

\end{document}